\documentclass[11pt, oneside]{article}   	
\usepackage[margin=1in]{geometry}            		
\usepackage{graphicx}	
\usepackage{caption}
\usepackage{subcaption}				
\usepackage{amssymb}
\usepackage{enumerate}
\usepackage{amsmath}
\usepackage{amsthm}
\usepackage{xcolor}
\usepackage{hyperref}
\hypersetup{
    linktoc=all,     
    linkcolor=blue,  
}
\usepackage[all]{xy}

\newtheorem{theorem}{Theorem}[section]
\newtheorem{lemma}[theorem]{Lemma}
\newtheorem{proposition}[theorem]{Proposition}
\newtheorem{corollary}[theorem]{Corollary}

\theoremstyle{definition}
\newtheorem{definition}[theorem]{Definition}
\theoremstyle{definition}
\newtheorem{example}[theorem]{Example}

\newtheorem{remark}{Remark}

\newcommand{\C}{\mathcal C}
\newcommand{\D}{\mathcal D}
\newcommand{\E}{\mathcal E}
\newcommand{\F}{\mathbb F}
\newcommand{\R}{\mathbb R}
\newcommand{\U}{\mathcal U} 

\title{Neural ring homomorphisms and maps between neural codes}
\author{Carina Curto \& Nora Youngs}
\begin{document}

\maketitle

\begin{abstract} 
Neural codes are binary codes that are used for information processing and representation in the brain. In previous work, we have shown how an algebraic structure, called the {\it neural ring}, can be used to efficiently encode geometric and combinatorial properties of a neural code \cite{neuralring}. In this work, we consider maps between neural codes and the associated homomorphisms of their neural rings. In order to ensure that these maps are meaningful and preserve relevant structure, we find that we need additional constraints on the ring homomorphisms. This motivates us to define {\it neural ring homomorphisms}. Our main results characterize all code maps corresponding to neural ring homomorphisms as compositions of 5 elementary code maps. As an application, we find that neural ring homomorphisms behave nicely with respect to convexity. In particular, if $\C$ and $\D$ are convex codes, the existence of a surjective code map $\C\rightarrow  \D$ with a corresponding neural ring homomorphism implies that the minimal embedding dimensions satisfy $d(\D) \leq d(\C)$.
\end{abstract}

\tableofcontents

\section{Introduction}

A major challenge of mathematical neuroscience is to determine how the brain processes and stores information. By recording the spiking from a population of neurons, we obtain insights into their coding properties.  A {\it neural code} on $n$ neurons is a subset $\C\subset\{0,1\}^n$, with each binary vector in $\C$ representing an on-off pattern of neural activity. This type of neural code is referred to in the neuroscience literature as a combinatorial neural code \cite{Bialek2008,BialekBerry} as it contains only the combinatorial information of which neurons fire together, ignoring precise spike times and firing rates. These codes can be analyzed to determine important features of the neural data, using tools from coding theory \cite{Walker2013} and topology \cite{gap, CurtoBulletin}. 

A particularly interesting kind of neural code arises when neurons have {\it receptive fields}. These neurons are selective to a particular type of stimulus; for example, place cells respond to the animal's spatial location \cite{OKeefeDostrovsky}, and orientation-tuned neurons in visual cortex respond to the orientation of an object in the visual field \cite{hubelwiesel}.  The neuron's {\it receptive field} is the specific subset of the stimulus space to which that neuron is particularly sensitive, and within which the neuron exhibits a high firing rate. If all receptive fields for a set of neurons is known, one can infer the expected neural code by considering the overlap regions formed by the receptive fields. Figure \ref{fig:receptivefields} shows an arrangement of receptive fields, and gives the corresponding neural code.

\begin{figure}[h] %  figure placement: here, top, bottom, or page
   \centering
   \includegraphics[width=3.2in]{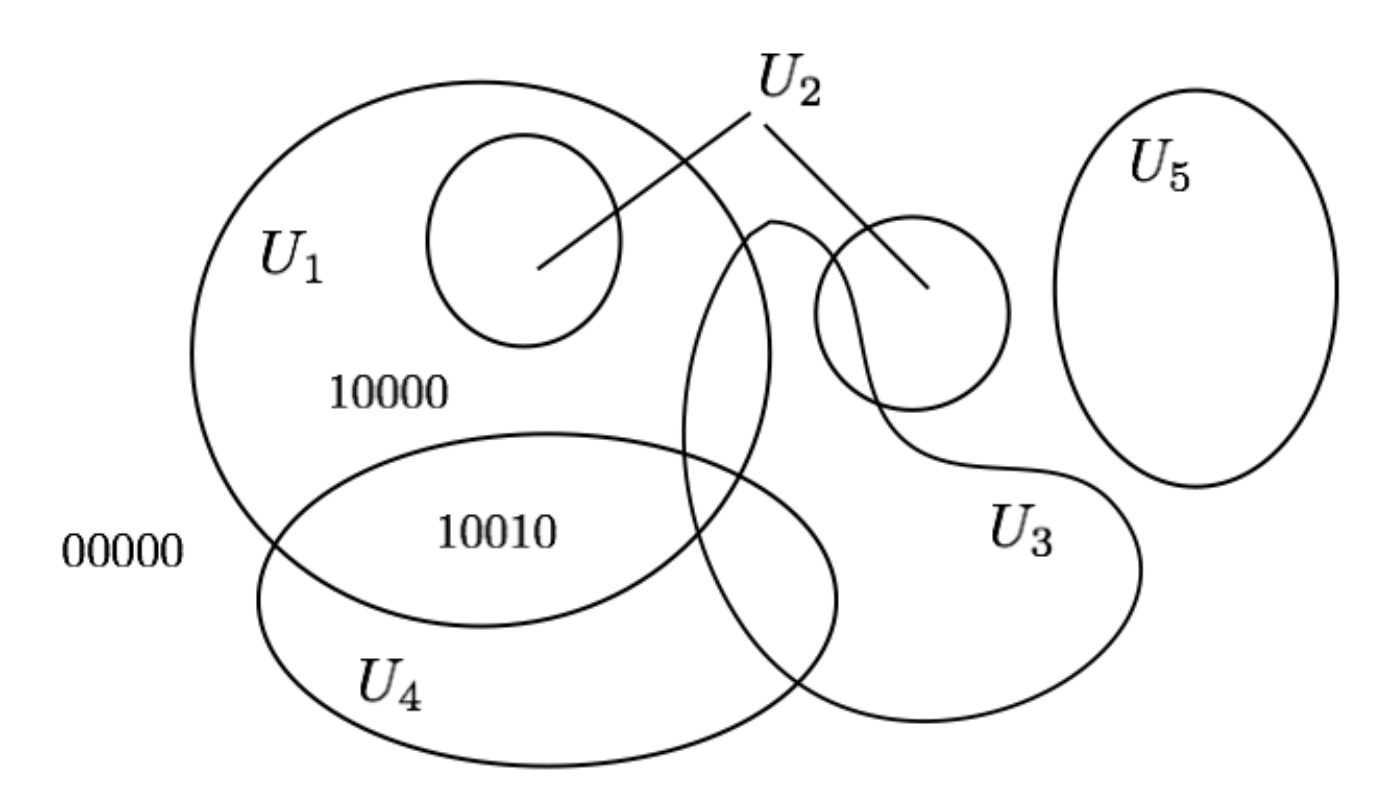} 
   \caption{An arrangement of five receptive fields $U_1,...,U_5$ in a stimulus space. Here, $U_i$ represents the receptive field of neuron $i$. The full code for the arrangement is: $\C = \{00000, 10000, 01000, 00100, 00010, 00001, 11000, 10100, 10010, 01100, 00110, 10110\}$.  }
   \label{fig:receptivefields}
\end{figure}

An arrangement of receptive fields whose regions correspond precisely to the neural code $\C$ is called a {\it realization} of $\C$. If the receptive fields can be chosen to be convex, then $\C$ is a {\it convex} neural code. Many neural codes are observed to be convex \cite{YartsevUlanovsky, CurtoBulletin}. In this case, we can leverage results from the extensive literature on arrangements of convex sets, such as Helly's theorem \cite{helly-review}, to give bounds on the dimension of the space of stimuli (see \cite{neuralring} for some examples). Note that the code in Figure \ref{fig:receptivefields} is convex, even though the realization depicted there is not; it is easy to see that $U_3$ can be redrawn as a convex set without altering the code.

 In previous work \cite{neuralring}, we introduced the neural ideal and the corresponding neural ring, algebraic objects associated to a neural code that capture its combinatorial properties. Thus far, work involving the neural ring has been primarily concerned with using the algebraic framework to extract structural information about the code \cite{neuralring, MRCalgsigs} and to determine which codes have convex realizations \cite{MRCpaper}. However, a neural code $\C$ is not an isolated object, but rather a member of a family of codes.  We define a {\it code map} from a code $\C$ to another code $\D$  to be any well-defined function $q:\C\rightarrow \D$.  A code map may preserve important structural properties of a code, or it may completely ignore them and just send codewords to codewords in an arbitrary manner. We are interested in a set of `nice' code maps that reflect meaningful relationships between the corresponding neural codes.  Our primary motivating examples of `nice' maps are those which leave the structure of a code essentially intact:
 
\begin{enumerate}
\item {\bf Permutation:} If $\C$ and $\D$ are identical codes up to a re-ordering of the neurons, then the permutation map $q: \C \rightarrow \D$ is `nice,' as it perfectly preserves the combinatorial structure.
\item {\bf Adding or removing trivial neurons:} A code $\C$ can be trivially changed by appending an extra neuron that has uniform behavior in all codewords -- i.e., always silent or always firing. Similarly, a code that has a neuron which is always ``on" or always ``off" is structurally equivalent to the code obtained by removing these trivial neurons, and the corresponding maps are `nice.'
\end{enumerate}

One way to obtain a code with trivial neurons is via localization. For example, consider the code in Figure \ref{fig:receptivefields}, restricted to the codewords whose regions are all contained inside $U_1$.  This code has five codewords: $\C' = \{10000, 11000, 10100, 10010, 10110\}$.  There is a natural map $q:\C' \rightarrow \D$ that drops neurons 1 and 5, which are both trivial, to obtain $\D = \{000, 100, 010, 001, 011\}$, which is structurally equivalent to $\C'$. 
Not all code maps respect the structure of the corresponding codes, however. For example, there is no guarantee that an arbitrary code map $\C' \rightarrow \D$ will reflect the fact that these codes are structurally equivalent.

%\begin{example} Let $\C = \{00, 10, 11\}$ and $\D = \{00, 01, 11\}$. These two codes could be related by the simple permutation $q_1$ that swaps neuron $1$ and neuron $2$, so $q_1(00)=00$, $q_1(01)=10$, and $q_1(11)=11$. This shows the essential similarity of the codes by simply trading the behavior of neurons $1$ and $2$.
%
%However, we could also define the code map $q_2:\C\rightarrow \D$ by $q_2(00) = q_2(11)= 11$, $q_2(01)=10$. The behavior of neuron 1 or 2 has not been preserved by this mapping. 
%
%\end{example}

%THIS TRANSITION NEEDS WORK

In this article, we consider how maps between neural codes relate to neural rings, as first defined in \cite{neuralring}. Our main questions are, simply:\\

\noindent{\bf Questions.} What types of maps between neural rings should be considered `nice'? How should we define neural ring homomorphisms? What other code maps correspond to nice maps between the associated neural rings?\\

These questions are analogous to studying the relationship between maps on algebraic varieties and their associated rings \cite{cox-little-oshea}.  However, as we will see in the next section, the standard notions of ring homomorphism and isomorphism are much too weak to capture any meaningful structure in the related codes.
Recent work \cite{MoAmzi} considered which ring homomorphisms preserve neural ideals as a set, and described corresponding transformations to codes through that lens.  In this article, we will define a special class of maps, called {\it neural ring homomorphisms}, that capture the structure of the nice code maps described above, and also guide us to discover additional code maps which should be considered `nice.' Our main result, Theorem~\ref{thm:mainthm}, characterizes all code maps that correspond neural ring homomorphisms and isomorphisms as compositions of five elementary code maps (including the two `nice' types above). As an application, Theorem~\ref{thm:convexity} shows that any surjective code map with a corresponding neural ring homomorphism preserves convexity and can only lower the minimal embedding dimension.
 
The organization of this paper is as follows. In Section 2, we review the neural ring of a code and describe the relevant pullback map, which gives a correspondence between code maps and ring homomorphisms. This allows us to see why the usual ring homomorphisms between neural rings are insufficiently restrictive. In Section \ref{sec:nrhoms} we define {\it neural ring homomorphisms}, a special class of maps that preserve code structure, and state Theorem~\ref{thm:mainthm}.  In Section 3.1 we take a closer look at the new elementary code maps that emerged in Theorem~\ref{thm:mainthm}, and  we prove the theorem. Finally, in Section \ref{sec:convexity}, we state and prove Theorem \ref{thm:convexity}, showing that surjective code maps corresponding to neural ring homomorphisms are particularly well-behaved with respect to convexity.

\section{Neural rings and the pullback map}\label{sec:pullback}

First, we briefly review the definition of a neural code and its associated neural ring, as previously defined in \cite{neuralring}. We then present the pullback map, which naturally relates maps between codes to homomorphisms of neural rings.

\begin{definition} A {\it neural code} on $n$ neurons is a set of binary firing patterns of length $n$.  Given neural codes $\C\subset\{0,1\}^n$ and $\D\subset\{0,1\}^m$, on $n$ and $m$ neurons, a {\it code map} is any function $q:\C\rightarrow \D$ sending each codeword $c\in \C$ to another codeword $q(c) \in \D$. 

For any neural code $\C\subset \{0,1\}^n$, we define the associated ideal $I_\C\subset\F_2[x_1,...,x_n]$ as follows: $$I_\C \stackrel{\text{def}}{=} \{f\in \F_2[x_1,...,x_n] \ | \ f(c) = 0 \ \text{ for all } \ c\in \C\}.$$ 
The {\it neural ring} $R_\C$ is then defined to be $R_\C = \F_2[x_1,...,x_n]/I_\C$. 
\end{definition} 
 Note that the neural ring $R_\C$ is precisely the ring of functions $\C\rightarrow \{0,1\}$, denoted $\F_2^\C$.  Since the ideal $I_\C$ consists of polynomials that vanish on $\C$, we can make use of the ideal-variety correspondence to obtain an immediate relationship between code maps and ring homomorphisms by using the pullback map. Given a code map $q:\C\rightarrow \D$, each $f\in R_\D$ is a function $f:\D\rightarrow \{0,1\}$, and therefore we may ``pull back" $f$ by $q$ to a function $f\circ q:\C\rightarrow \{0,1\}$, which is an element of $R_\C$.  Hence, for any $q:\C\rightarrow \D$, we may define the pullback map $q^*:R_\D\rightarrow R_\C$, where $q^*(f) = f\circ q$, as illustrated below:
 \[
 \xymatrix{ \C \ar[dr]_{q^*f = f\circ q} \ar[r]^q & \D\ar[d]^f\\  & \{0,1\}}
 \]

It is easy to check that for any code map $q:\C\rightarrow \D$, the pullback $q^*:R_\D\rightarrow R_\C$ is a ring homomorphism. In fact, the pullback provides a bijection between code maps and ring homomorphisms, as the following proposition states.

\begin{proposition}\label{thm:bijection}
There is a 1-1 correspondence between code maps $q:\C\rightarrow \D$ and ring homomorphisms $\phi:R_\D\rightarrow R_\C$, given by the pullback map.  That is, given a code map $q:\C\rightarrow \D$, its pullback $q^*:R_\D\rightarrow R_\C$ is a ring homomorphism; conversely, given a ring homomorphism $\phi:R_\D\rightarrow R_\C,$ there is a unique code map $q_\phi:\C\rightarrow \D$ such that $q_\phi^* = \phi$. 
\end{proposition}

Proposition \ref{thm:bijection} is a special case of \cite[Proposition 8, p. 234]{cox-little-oshea}. We will include our own proof in Section \ref{sec:pullbackpfs} in order to show constructively how to obtain $q_\phi$ from $\phi$.
Unfortunately, Proposition \ref{thm:bijection} also makes it clear that ring homomorphisms $R_\D \rightarrow R_\C$ need not preserve any structure of the associated codes, as {\it any} code map has a corresponding ring homomorphism. 
The next proposition tells us that even ring {\it isomorphisms} are quite weak: any pair of codes with the same number of codewords admits an isomorphism between the corresponding neural rings.

%ISOMORPHISM IS NOT THE ANSWER
\begin{proposition}\label{prop:iso} A ring homomorphism $\phi:R_\D\rightarrow R_\C$ is an isomorphism if and only if the corresponding code map $q_\phi:\C\rightarrow \D$ is a bijection. 
\end{proposition}

 Propositions~\ref{thm:bijection} and \ref{prop:iso} highlight the main difficulty with using ring homomorphism and isomorphism alone: the neural rings are rings of functions from $\C$ to $\{0,1\}$, and the abstract structure of such a ring depends solely on the number of codewords, $|\C|$. Considering such rings abstractly, independent of their presentation, reflects no additional structure -- not even the code length (or number of neurons, $n$) matters.  In particular, we cannot track the behavior of the variables $x_i$ that represent individual neurons. This raises the question: what algebraic constraints can be put on homomorphisms between neural rings in order to capture a meaningfully restricted class of code maps?

\subsection{The pullback correspondence: a closer look.}\label{sec:pullbackpfs}

 Before moving on to defining a more restricted class of homomorphisms, we introduce some notation to take a closer look at neural rings, and how the correspondence between code maps and homomorphisms occurs. Using this, we provide concrete and elementary proofs of Propositions~\ref{thm:bijection} and \ref{prop:iso}.
 
   Elements of neural rings may be denoted in different ways. First, they can be written as polynomials, where it is understood that the polynomial is a representative of its equivalence class mod $I_\C$. Alternatively, using the vector space structure, an element of $R_\C$ can be written as a function $\C\rightarrow \{0,1\}$ defined completely by the codewords that support it.  We will make use of the latter idea frequently, so it is helpful to identify a canonical basis of characteristic functions $\{\rho_c\,|\,c\in \C\}$, where $$\rho_c(v) = \left \{\begin{array}{ll} 1 &\text{if}\;\; v=c , \\ 0 & \text{otherwise.}\end{array}\right.$$  In polynomial notation, 
   $$\rho_c=\prod_{c_i = 1} x_i \prod_{c_j = 0} (1-x_j),$$ 
where $c_i$ represents the $i$th component of codeword $c$.  The characteristic functions $\rho_c$ form a basis for $R_\C$ as an $\F_2$-vector space, and they have several useful properties:
\begin{itemize}
\item Each element $f$ of $R_\C$ can be represented as the formal sum of basis elements for the codewords in its support:  $\displaystyle f=\sum_{\{c \in \C \mid f(c) = 1\}} \rho_c$.
\item In particular, we can write $x_i = \displaystyle\sum_{\{c\in \C\,\mid \,c_i=1\}}\rho_c$. So, if $c_i=c_j$ for all $c \in \C$, then $x_i = x_j$.  Likewise, if $c_i=1$ for all $c\in \C$, we have $x_i=1$.
\item  The product of two basis elements is 0 unless they are identical: $\rho_c\rho_d = \left\{\begin{array}{ll} \rho_c  &\text{if}\;\; c=d ,\\ 0 & \text{ otherwise }\end{array}\right.$.
\item If $1_\C$ is the identity of $R_\C$, then $\displaystyle 1_\C = \sum_{c\in \C} \rho_c$.
\end{itemize}

%THE PULLBACK MAP 

Once we have a homomorphism $\phi:R_\D\rightarrow R_\C$, we necessarily have a map which sends basis elements of $R_\D$ to sums of basis elements in $R_\C$. We will now show how this illustrates the corresponding code map. First, a technical lemma.
 
 %DEFINING CODE MAPS FROM PHI
 \begin{lemma}\label{lem:technical} For any ring homomorphism $\phi:R_\D \rightarrow R_\C$, and any element $c\in\C$, there is a unique $d\in \D$ such that $\phi(\rho_{d})(c) = 1$.
\end{lemma} 
 
 \begin{proof}
To prove existence, note that $\sum_{c\in \C} \rho_c = 1_{\C} = \phi(1_{\D}) = \phi(\sum_{d\in \D} \rho_d)  = \sum_{d\in \D} \phi(\rho_d)$. For each $c\in \C$, $1 = \rho_c(c) =  \left(\sum_{c'\in \C} \rho_{c'}\right)(c) = \left(\sum_{d\in \D} \phi(\rho_d)\right)(c)$, and thus 
$\phi(\rho_d)(c)=1$ for at least one $d\in \D$. 
To prove uniqueness, suppose there exist distinct $d,d'\in \D$ such that $\phi(\rho_d)(c) = \phi(\rho_{d'})(c) = 1$.  Then as $\phi$ is a ring homomorphism, we would have $1= (\phi(\rho_d)\phi(\rho_{d'}))(c) = \phi(\rho_d\rho_{d'})(c) = \phi(0)(c) = 0$, but this is a contradiction.  Thus such a $d$ must be unique. 
\end{proof}

 This result allows us to describe the unique code map corresponding to any ring homomorphism.

 \begin{definition}
 Given a ring homomorphism $\phi:R_\D\rightarrow R_\C$, we define the associated code map $q_\phi:\C\rightarrow \D$ as follows: 
  $$q_\phi(c) = d_c$$ where $d_c$ is the unique element of $\D$ such that $\phi(\rho_{d_c})(c) =1$, guaranteed by Lemma \ref{lem:technical}.  
  \end{definition}
  
Using this definition, we are able to prove Proposition \ref{thm:bijection}.

  \begin{proof}[Proof of Proposition \ref{thm:bijection}] It is easy to check that the pullback $q^*$ is a ring homomorphism; we now prove that any homomorphism can be obtained as the pullback of a code map.  Given a ring homomorphism $\phi:R_\D \rightarrow R_\C$, define $q_\phi$ as above. We must show that the $q_\phi^* = \phi$, and moreover that $q_\phi$ is the only code map with this property.  

The fact that  $q_\phi^* = \phi$ holds essentially by construction: let $f\in R_\D$, so $f=\sum_{f(d) = 1} \rho_d$.  Then, for any $c\in \C$, 
$$q_\phi^*(f)(c) = f(q_\phi(c)) = \sum_{f(d) = 1} \rho_d(q_\phi(c)) = \sum_{f(d)=1}\rho_d(d_c) = \left\{\begin{array}{ll} 1 & \text{ if } f(d_c) = 1\\ 0 & \text{ if }f(d_c) = 0 \end{array}\right.$$
whereas, remembering from above that there is exactly one $d\in \D$ such that $\phi(\rho_d)(c)=1$ and that this $d$ may or may not be in the support of $f$, we have

$$\phi(f)(c) = \sum_{f(d)=1} \phi(\rho_d)(c)=  \left\{\begin{array}{ll} 1 & \text{ if }d_c\in f^{-1}(1)  \\ 0 & \text{ if } d_c\notin f^{-1}(1) \end{array} \right. = \left\{\begin{array}{ll} 1 & \text{ if }f(d_c) = 1 \\ 0 & \text{ if }f(d_c) = 0 \end{array} \right. . $$
Thus, $\phi = q_\phi^*$.  

Finally, to see that $q_\phi$ is the only code map with this property, suppose we have a different map $q\neq q_\phi$.  Then there is some $c\in \C$ with $q(c) \neq q_\phi(c)$; let $d_c = q_\phi(c)$, so $q(c)\neq d_c$.  Then $\phi(\rho_{d_c})(c) = 1$ by definition, but $q^*(\rho_{d_c})(c) = \rho_{d_c}(q(c)) =0$ as $q(c)\neq d_c$.  So $q^*$ does not agree with $\phi$ and hence $\phi$ is not the pullback of $q$, so $q_\phi$ is the unique code map with pullback $\phi$.
\end{proof}

The following example illustrates the connection between a homomorphism $\phi$ and the corresponding code map $q_\phi$.

\begin{example}   Let $\C = \{110, 111, 010, 001\}$ and $\D = \{00, 10, 11\}$.  Let $\phi:R_\D\rightarrow R_\C$ be defined by $\phi(\rho_{11} )= \rho_{110} + \rho_{111} + \rho_{010}$, $\phi(\rho_{00}) = \rho_{001}$, and $\phi(\rho_{10}) = 0$.   Then the corresponding code map $q_\phi$ will have $q_\phi(110) = q_\phi(111) = q_\phi(010) = 11$, and $q_\phi(001) = 00$. Note that there is no element $c\in \C$ with $q_\phi(c) = 10$ so $q_\phi$ is not surjective.
\end{example}

Finally, we provide a proof of Proposition \ref{prop:iso}. 

\begin{proof}[Proof of Proposition \ref{prop:iso}]
Note that $R_\C\cong \F_2^{|\C|}$ and $R_\D \cong \F_2^{|\D|}$, and $\F_2^{|\C|} \cong \F_2^{|\D|}$ if and only if $|\C| = |\D|$.   Suppose $\phi$ is an isomorphism; then we must have $|\C| = |\D|$.  If $q_\phi$ is not injective, then there is some $d\in \D$ such that $\phi(\rho_d)(c) = 0$ for all $c\in \C$.  But then $\phi(\rho_d) = 0$, which is a contradiction since $\phi$ is an isomorphism so $\phi^{-1}(0) = \{0\}$.  Thus $q_\phi$ is injective, and since $|\C| = |\D|$, this means $q_\phi$ is a bijection.

On the other hand, suppose $q_\phi:\C\rightarrow \D$ is a bijection. Then $|\C| = |\D|$, so $R_\C\cong R_\D$, and as both are finite, $|R_\C|= |R_\D|$.  Consider an arbitrary element $f\in R_\C$.   For each $c\in f^{-1}(1)$, there is a unique $d\in \D$ so $\phi(\rho_d) = c$; furthermore as $q_\phi$ is a bijection, all these $d$ are distinct.  Then $$\phi\big(\sum_{\substack{d=q_\phi(c),\\ c\in f^{-1}(1)}} \rho_d\big) = \sum_{\substack{d=q_\phi(c)\\c\in f^{-1}(1)}} \phi(\rho_d) = \sum_{c\in f^{-1}(1)} \rho_c = f. $$ Hence $\phi$ is surjective, and since $|R_\C| = |R_\D|$, $\phi$ is also bijective and hence an isomorphism.
\end{proof}

\section{Neural ring homomorphisms}\label{sec:nrhoms}

In order to define a restricted class of ring homomorphisms that preserve certain structural similarities of codes, we consider how our motivating maps (permutation and adding or removing trivial neurons) preserve structure. In each case, note that the code maps act by preserving the activity of each neuron: we do not combine the activity of neurons to make new ones, or create new neurons that differ in a nontrivial way from those we already have. Following this idea, we restrict to a class of maps that respect the elements of the neural ring corresponding to individual neurons: the variables $x_i$. Here we use the standard notation $[n]$ to denote the set $\{1,\ldots,n\}$.

\begin{definition} Let $\C\subset\{0,1\}^n$ and $\D\subset \{0,1\}^m$ be neural codes, and let $R_\C = \F_2[y_1,...,y_n]/I_\C$ and $R_\D = \F_2[x_1,...,x_m]/I_\D$ be the corresponding neural rings.  A ring homomorphism $\phi:R_\D\rightarrow R_\C$ is a {\it neural ring homomorphism} if  $\phi(x_j)\in\{ \{y_i\,|\, i\in[n]\}, 0,1\}$ for all $j\in [m]$.  
We say that a neural ring homomorphism $\phi$ is a {\it neural ring isomorphism} if it is a ring isomorphism and its inverse is also a neural ring homomorphism. 
\end{definition}

It is important to remember that when we refer to the `variables' of $R_\D$, we actually mean the {\it equivalence class} of the variables under the quotient ring structure. Thus, it is possible in some cases to have $x_i = x_j$, or $x_i=0$, depending on whether these variables give the same function on all codewords. We now provide some examples to illustrate neural ring homomorphisms.

\begin{example}\label{ex:nrhoms}  Here we consider three different code maps: one that corresponds to a neural ring isomorphism, one that corresponds to a neural ring homomorphism but not to a neural ring isomorphism, and one that does not correspond to a neural ring homomorphism at all.

\begin{enumerate}

\item Let $\D = \{0000, 1000, 0001, 1001, 0010, 1010, 0011\}$, and let \\ $\C=
\{0000,0001, 0010, 0011, 0100, 0101, 0110\}$.  Define $\phi:R_\D\rightarrow R_\C$ as follows: \[
\begin{array}{cc}
\phi(\rho_{0000}) = \rho_{0000}  & \phi(\rho_{1000}) = \rho_{0001}\\
\phi(\rho_{0001}) = \rho_{0010} & \phi(\rho_{1001}) = \rho_{0011} \\ 
\phi(\rho_{0010}) = \rho_{0100} & \phi(\rho_{1010}) = \rho_{0101}\\ 
\phi(\rho_{0011}) = \rho_{0110} & \\
\end{array}
\]

 Note that  $\phi(x_1) = \phi(\rho_{1000} + \rho_{1010}) = \rho_{0001} + \rho_{0101} = y_4$, and $\phi(x_2) = \phi(0) = 0 = y_1$. By similar calculations, we have $\phi(x_3)=y_2$, and $\phi(x_4) = y_3$.  Thus, $\phi$ is a neural ring homomorphism;  in fact, since $\phi$ is a ring isomorphism and its inverse is a neural ring homomorphism sending $\phi^{-1}(y_1)=0=x_2$, $\phi^{-1}(y_2)=x_3$, $\phi^{-1}(y_3)=x_4$, and $\phi^{-1}(y_4)=x_1$, $\phi$ is a neural ring isomorphism..

\item Let $\D = \{000,110\}$ and $\C = \{00, 01,10\}$. Define $\phi:R_\D\rightarrow R_\C$ by $\phi(\rho_{000}) = \rho_{00} + \rho_{10}$ ad  $\phi(\rho_{110}) = \rho_{01}$. In $R_\D$, $x_1=x_2 = \rho_{110}$, and $x_3=0$. In $R_\C$, we have $y_1=\rho_{10}$ and $y_2 = \rho_{01}$. Under this map, we find $\phi(x_1) = \phi(x_2) = y_1$ and $\phi(x_3) = 0$, so $\phi$ is a neural ring homomorphism. However, it is not a neural ring isomorphism, as it is not a ring isomorphism.

\item Let $\D = \{00, 10\}$ and $\C = \{00, 10, 01\}$. Define the ring homomorphism $\phi:R_\D\rightarrow R_\C$ as follows: $\phi(\rho_{00}) = \rho_{00}$, $\phi(\rho_{10}) = \rho_{10} + \rho_{01}$. In $R_\D$, $x_1=\rho_{10}$. However, $\phi(x_1) = \rho_{10}+\rho_{01}$, which is not equal to either $y_1=\rho_{10}$, $y_2 = \rho_{01}$, $1=\rho_{00}+\rho_{10} + \rho_{01}$, or $0$. Thus, $\phi$ is not a neural ring homomorphism.

\end{enumerate} 

\end{example}

It is straightforward to see that the composition of neural ring homomorphisms is again a neural ring homomorphism.

\begin{lemma} If $\phi:R_\D\rightarrow R_\C$ and $\psi:R_\E\rightarrow R_\D$ are neural ring homomorphisms, then their composition $\phi\circ \psi$ is also a neural ring homomorphism. If $\phi$ and $\psi$ are both neural ring isomorphisms, then their composition $\phi\circ \psi$ is also a neural ring isomorphism.
\end{lemma}

%Proof of "3.3" %\begin{proof} Since the composition of two ring homomorphisms (isomorphisms) is a ring homomorphism (isomorphism), we need only show that the variables are mapped as specified by the definition. Let $\{z_k\,|\, k\in [p]\}$ be the variables for $R_\E$, $\{y_j\,|\, j\in [m]\}$ be the variables for $R_\D$, and $\{x_i\,|\, i\in [n]\}$ be the variables for $R_\C$. Since $\psi$ is a neural ring homomorphism, we have $\psi(\{z_k\}, 0, 1) \subset \{ \{y_j\}, 0,1\}$ and since $\phi$ is a neural ring homomorphism, we have $\phi(\{y_j\}, 0,1) \subset \{\{x_i\}, 0,1\}$, so in the composition we have $\phi \circ \psi(\{z_k\}, 0,1) \subset \{\{x_i\}, 0,1\}$ and thus the composition is also a neural ring homomorphism. Similarly, if $\phi$ and $\psi$ are both neural ring isomorphisms, then there is a surjection of variables, so we have $\psi(\{z_k\}, 0, 1) = \{ \{y_j\}, 0,1\}$ and $\phi(\{y_j\}, 0,1) = \{\{x_i\}, 0,1\}$, and thus in the composition $\phi\circ \psi(\{z_k\}, 0,1) =\{\{x_i\}, 0,1\}$ and the composition is a neural ring isomorphism.
%\end{proof}

%\subsection{Main Results}\label{sec:results}

As we have seen in Example \ref{ex:nrhoms}, both permutations and appending a trivial neuron correspond to neural ring isomorphisms. The following theorem introduces three other types of elementary code maps, which yield neural ring homomorphisms. All of these code maps are meaningful in a neural context, and preserve the behavior of individual neurons.  And, as seen in Theorem~\ref{thm:mainthm}, it turns out that {\it all} neural ring homomorphisms correspond to code maps that are compositions of these five elementary types of maps. The proof is given in Section \ref{sec:nrhpfs}. 

\begin{theorem}\label{thm:mainthm} A map $\phi:R_\D\rightarrow R_\C$ is a neural ring homomorphism if and only if $q_\phi$ is a composition of the following elementary code maps:
\begin{enumerate}
\item Permutation 
\item Adding a trivial neuron (or deleting a trivial neuron)
\item Duplication of a neuron (or deleting a neuron that is a duplicate of another)
\item Neuron projection (deleting a not necessarily trivial neuron)
\item Inclusion (of one code into another)
\end{enumerate}
Moreover, $\phi$ is a neural ring isomorphism if and only if $q_\phi$ is a composition of maps (1)-(3).
\end{theorem}

The ability to decompose any `nice' code map into a composition of these five elementary maps has immediate consequences for answering questions about neural codes. For example, one of the questions that motivated the definition of the neural ring and neural ideal was that of determining which neural codes are {\it convex}. In Section \ref{sec:convexity}, we look at how each of these maps affect convexity.

The following example provides a sense of what these different operations mean.

\begin{example}\label{ex:fivemaps}  In Figure \ref{fig:codes} we show a code $\C$, and the resulting codes $\C_1,\ldots,\C_5$ after applying the following elementary code maps: 
\begin{enumerate}
\item the cyclic permutation (1234) ($\C_1$), 
\item adding a trivial always-on neuron ($\C_2$), 
\item duplication of neuron 4 ($\C_3$),
\item   deleting neuron 4 (projecting onto neurons 1-3) ($\C_4$) 
\item an inclusion map into a larger code ($\C_5$). 
\end{enumerate}
The effects of these code maps on a realization of $\C$ are  shown on the left of Figure \ref{fig:codes}. The succeeding columns in the table on the right give the image of $\C$ under each of the five code maps. \\

%FIGURE!!!

\begin{figure}[!h]
\centering
\begin{subfigure}{.6\textwidth}
  \centering
  \includegraphics[width=3.2in]{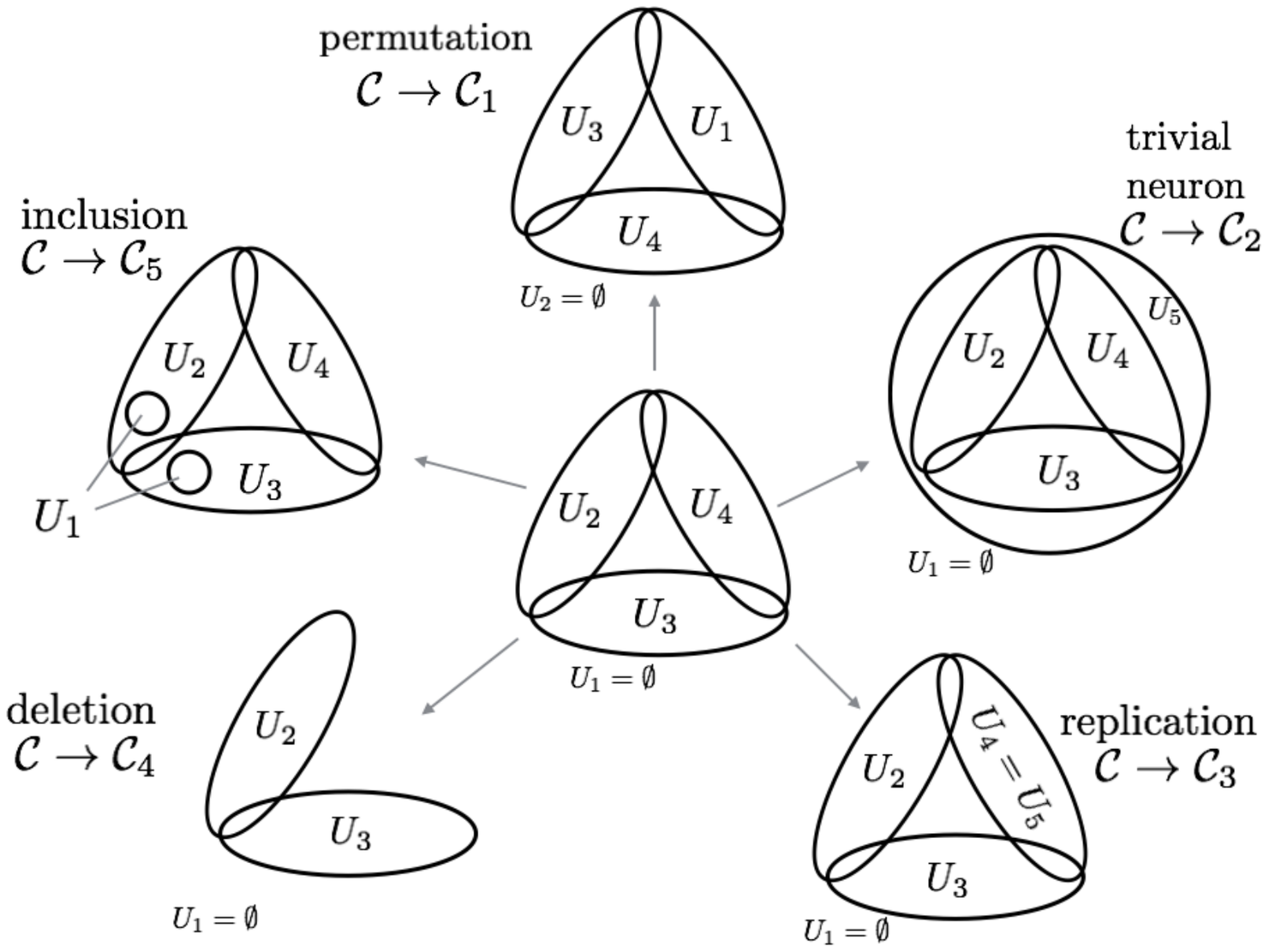}
\end{subfigure}%
\begin{subfigure}{.4\textwidth}
  \centering
  \includegraphics[width=2.8in]{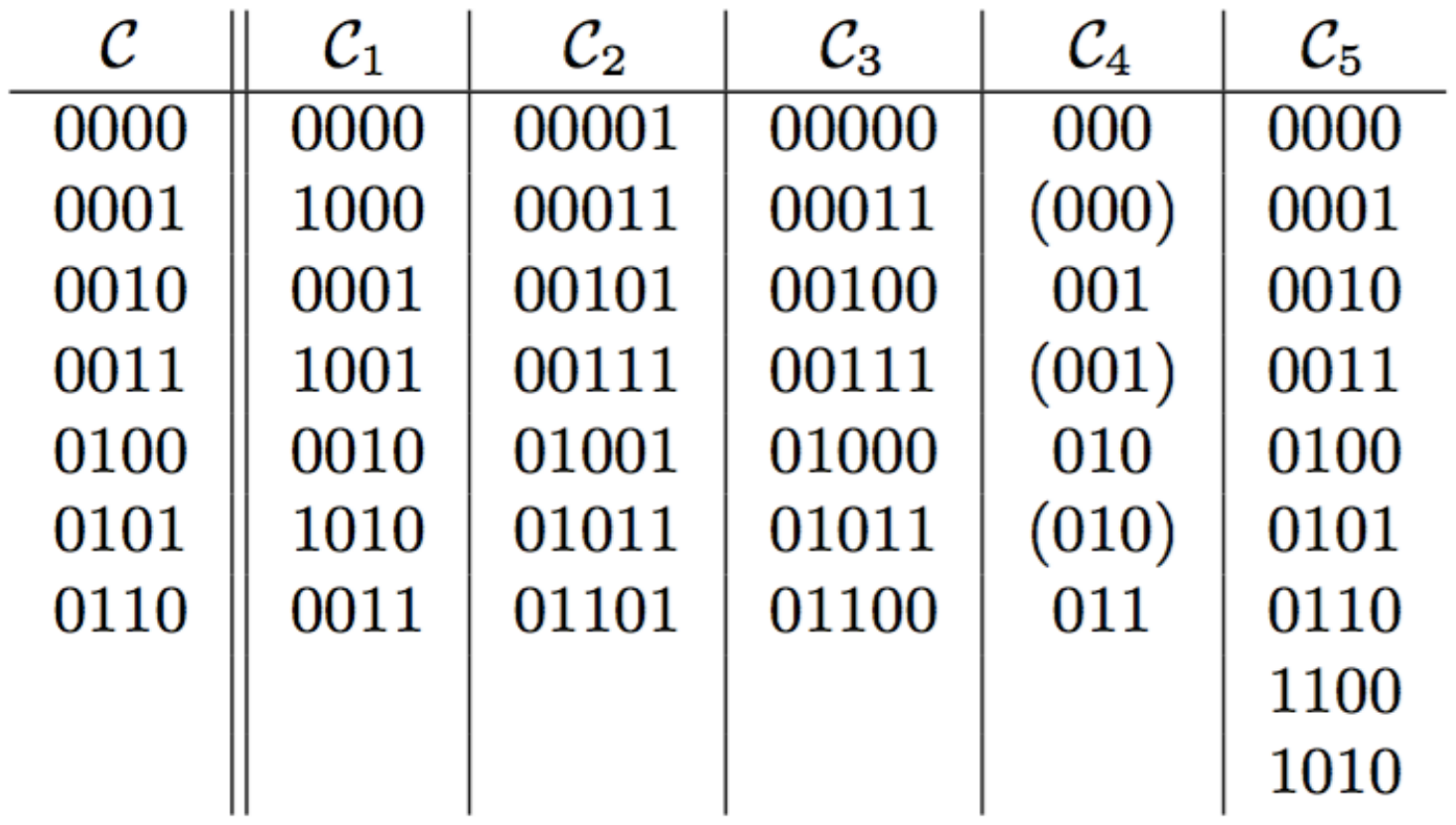}
\end{subfigure}
\caption{A code $\C$ and its image under five elementary code maps. (Left) The effect of each codeword on a realization of $\C$. (Right) A table showing how each codeword of $\C$ is transformed by each map. In each case, the code map sends a codeword $c\in \C$ to the codeword in its row.}
\label{fig:codes}
\end{figure}

\end{example}

\subsection{Proof of Theorem \ref{thm:mainthm} }\label{sec:nrhpfs}

To prove Theorem \ref{thm:mainthm}, we will first focus on the structure of neural ring homomorphisms.   As neural ring homomorphisms strictly control the possible images of variables, they can be described succinctly by an index `key' vector that captures the information necessary to determine the map. Since the index for the first variable will use the symbol `1', we will where necessary denote the multiplicative identity $1$ of the ring with the symbol $u$ to distinguish the two. Throughout, we will use the notation $c_i$ to indicate the $i$th component of a codeword $c$. 

\begin{definition} Let $\phi:R_\D\rightarrow R_\C$ be a neural ring homomorphism, where $\C$ and $\D$ are codes on $n$ and $m$ neurons, respectively. The {\it key vector} of $\phi$ is the vector $V\in \{1,...,n, 0, u\}^m$ such that $$V_j = \left\{ \begin{array}{ll} i & \text{ if }\phi(x_j) = y_i\\ 0 & \text{ if } \phi(x_j) = 0\\ u & \text{ if } \phi(x_j) = 1\end{array}\right..$$
\end{definition}

This key vector completely describes a neural ring homomorphism, since once the image of each variable is determined the rest of the homomorphism is given by the usual properties of homomorphism.  In cases where we have $y_i=y_k$ for some $i,k$, then only one representative of the equivalence class need appear in $V$.  

Because of the close correspondence of code maps and ring homomorphisms, the key vector also completely determines the associated code map. The following lemma gives the explicit relationship.

\begin{lemma} Let $\phi:R_\D\rightarrow R_\C$ be a neural ring homomorphism with key vector $V$. Then the corresponding code map $q_\phi:\C\rightarrow \D$ is given by $q_\phi(c) = d$, where $d_j = \left\{\begin{array}{ll} c_i & \text{ if }V_j=i\\ 0 &\text{ if } V_j =0\\ 1 & \text{ if } V_j=u\end{array}\right.$.
\end{lemma}

%\begin{proof} 
%This lemma holds directly from the pullback map correspondence. Let $\phi:R_\D\rightarrow R_\C$ be a neural ring homomorphism as specified, and so for every $x_j$, we have $\phi(x_j)\in \{\{y_i\}, 0, 1\}$. Denote $q_\phi(c) = d$, and recall that for any $f\in R_\D$, we must have $f(d) = (\phi\circ f)(c)$. Taking the function $f=x_j$ we must have $$d_j=x_j(d) = x_j(q_\phi(c)) = \phi(x_j)(c) = \left\{\begin{array}{ll} c_i & \text{ if } \phi(x_j) = y_i \\ 0 & \text{ if } \phi(x_j) = 0 \\ 1 &  \text{ if }\phi(x_j) = 1\end{array}\right..$$ Since $V_j=i$ if $\phi(x_j) = i$, $V_j=0$ if $\phi(x_j) =0$, and $V_j=u$ if $\phi(x_j) =1$, this gives the desired result.
%\end{proof}

Furthermore, any code map that aligns with a key vector must be associated to a neural ring homomorphism.

\begin{lemma}
Let $\C$ and $\D$ be codes on $n$ and $m$ neurons, respectively. Suppose $q:\C\rightarrow \D$ is a code map and $V \in \{1,...,n,0,u\}^m$ such that $q$ is described by $V$; that is, for all $c\in \C$, $q(c)=d$ where $d_j = \left\{\begin{array}{ll} c_i & \text{ if }V_j=i\\ 0 &\text{ if } V_j =0\\ 1 & \text{ if } V_j=u\end{array}\right.$. Then the associated ring homomorphism $\phi_q$ is a neural ring homomorphism with key vector $V$.
\end{lemma}

\begin{proof} Let $q$ be as described above, and $\phi_q$ the associated ring homomorphism. We will show that for $j\in[m]$, we have $\phi_q(x_j) = \left\{ \begin{array}{ll} x_i & \text{ if }V_j = i\\ 0 & \text{ if } V_j = 0\\ 1 & \text{ if }V_j= 1\end{array}\right.$ and thus that $\phi_q$ is a neural ring homomorphism with key vector $V$. We will examine the three options for $V_j$ separately.

First, suppose $V_j = i \in [n]$. Then for all $c\in \C$, we have $q(c)_j = c_i$, and thus that $x_j(q(c)) = c_i$. Hence, $x_j\circ q = y_i$, since both functions act the same on all codewords $c\in \C$. But by definition of the pullback map, $\phi(x_j) = x_j\circ q$, so $\phi(x_j) = y_i$.
Next, suppose $V_j=0$. Then for all $c\in C$ we have $q(c)_j=0$ and thus that $x_j(q(c)) = 0$. Hence, $x_j\circ q = 0$, since both functions act the same on all codewords $c\in \C$. But by definition of the pullback map, $\phi(x_j) = x_j\circ q$, so $\phi(x_j) = 0$ in this case. 

Finally, suppose $V_j=u$. Then for all $c\in C$ we have $q(c)_j=1$ and thus that $x_j(q(c)) = 1$. Hence, $x_j\circ q = 1$, since both functions act the same on all codewords $c\in \C$. But by definition of the pullback map, $\phi(x_j) = x_j\circ q$, so $\phi(x_j) = 1$ in this case. 
\end{proof}

\begin{remark} It is important to note here that the key vector for a particular code map may not be unique. In Example \ref{ex:nrhoms} (1), we saw an example of a permutation code map that could be described by key vector $(4,1,2,3)$. However, as $\phi(x_2) = y_1=0$, we could replace this key vector with $(4,0,2,3)$ and describe the same homomorphism.  In cases like these, either choice is valid. However, this does not mean that the corresponding homomorphism is not unique.
\end{remark}

Now that we have shown that neural ring homomorphisms (and their corresponding code maps) are precisely those determined by key vectors, we need only show the following:

\begin{itemize}
\item  All five code maps listed have key vectors.
\item  Any code map with a key vector can be written as a composition of these five maps.
\item  The first three code maps correspond precisely to neural ring isomorphisms.
\end{itemize}

To see that all five elementary code maps in Theorem~\ref{thm:mainthm} have key vectors, we simply exhibit the key vector for each. In the process, we will show that the first three maps correspond to neural ring isomorphisms.
To describe these code maps, we will consider an arbitrary word $c\in \C$, written as $c=c_1c_2\cdots c_n$, and describe the image $q(c) \in \D$. Throughout, $\C$ is a code on $n$ neurons and $\D$ is a code on $m$ neurons.

\begin{enumerate}
\item Permutation maps: If the code map $q:\C\rightarrow \D$ is a permutation map, then $n=m$, $q(\C) = \D$, and each codeword is permuted by the same permutation $\sigma$. That is, for each $c\in \C$, we know $q(c) = c_{\sigma(1)}c_{\sigma(2)}\cdots c_{\sigma(n)}$. In this case, the key vector is given by $V_j = \sigma(j)$. As permutation yields a bijection on codewords, and the inverse permutation also has a key vector, permutation maps correspond to neural ring isomorphisms.

\item Adding a trivial neuron to the end of each codeword: in this case, $m=n+1$ and $q(\C) = \D$. Consider first the case of adding a trivial neuron that is never firing  to the end of each codeword, so that $q:\C\rightarrow \D$ is described by $q(c) = c_1c_2\cdots c_n0$, and $q(\C) = \D$. The key vector is given by $V_j = j$ for $j\in [n]$ and $V_{n+1} = 0$.  Similarly, if we add a neuron that is always firing, so $q(c) = c_1\cdots c_n 1,$ then $V_j=j$ for $j\in [n]$ and $V_{n+1} = u$. Such a map will be a bijection; moreover, the reverse map (where we delete the trivial neuron at the end of each word) also has a key vector: $W_i=i$ for all $i\in[n]$. Thus, this map (and its inverse) correspond to neural ring isomorphisms.

\item Adding a duplicate neuron to the end of each codeword: in this case, $m=n+1$ and $q(\C) =\D$. If the new neuron $n+1$ duplicates neuron $i$, then the code map is given by $q(c) = c_1\cdots c_n c_i$, and the key vector is given by $V_j=j$ for $j\in [n]$ and $V_{n+1} = i$. Such a map will be a bijection on codewords, and moreover, the inverse code map corresponds to the key vector where $W_i=i$ for all $i\in[n]$, and so its inverse corresponds to a neural ring homomorphism. Thus, this map and its inverse correspond to neural ring isomorphisms.

\item Projection (deleting the last neuron): in this case, $m=n-1$ and $q(\C)=\D$. The code map is given by $q(c) = c_1\cdots c_{n-1}$ and we have the key vector $V_j=j$ for $j\in [n-1]$. 

This map corresponds to a neural ring isomorphism precisely when the deleted neuron is either trivial, or a duplicate of another neuron. If neither of these hold, then there are two possibilities: either the code map is not a bijection, in which case the corresponding ring homomorphism is not an isomorphism, or the code map is a bijection, but the inverse will not be a neural ring homomorphism, as $\phi^{-1}(y_{n+1})\notin\{x_1,...,x_m, 0, 1\}$.

%Note that in the case where neuron $n$ was always $0$ (or always 1), this is the inverse of adding a trivial neuron. Similarly, in the case where neuron $n$ was identical to neuron $i$, this is the inverse map of adding a duplicate neuron. In these cases, while $V$ will not explicitly have an entry $n$, it will be the case that $y_n = 0, y_n=1$ or $y_n=x_i$, and so the induced map on variables will still be surjective; hence, the code map will correspond to a neural ring isomorphism.

\item Inclusion: in this case, $m=n$, and we have $q(c) = c$ for all $c\in \C$. However, we do not demand $q(\C)=\D$. Since in this case each codeword maps to itself, we can use the key vector $V_j=j$ for $j\in [n]$.

\end{enumerate}

Finally, we prove the main substance of Theorem \ref{thm:mainthm}, which is that any code map corresponding to a neural ring homomorphism can be written as a composition of the five listed maps, and furthermore that any isomorphism requires only the first three.

\begin{proof}[Proof of Theorem \ref{thm:mainthm}]
Let $\C$ and $\D$ be codes on $n$ and $m$ neurons, respectively, and let $\phi:R_\D\rightarrow R_\C$ be a neural ring homomorphism with corresponding code map $q$. Our overall steps will be as follows: 

\begin{enumerate}
\item Append the image $q(c)$ to the end of each codeword $c$ using a series of maps that duplicate neurons or add trivial neurons, as necessary. 
\item Use a permutation map to move the image codeword $q(c)$ to the beginning, and the original codeword $c$ to the end.
\item Use a series of projection maps to delete the codeword $c$ from the end, resulting in only $q(c)$.
\item Use an inclusion map to include $q(\C)$ into $\D$ if $ q(\C)\subsetneq \D$.
\end{enumerate}

 %Using this key vector define the function $f_j\in \F_2[x_1,...,x_n]$ for $j=1,...,m$, such that $f_j= \left\{\begin{array}{ll} x_j & \text{ if } V_i = j\\ 1 & \text{ if }V_i = u \\ 0 & \text{ if }V_i=0\end{array}\right. $ 

First we define some intermediate codes: let $\C_0=\C$.  For $j=1,...,m$, let $$\C_j = \{ (c_1,...,c_n,d_1,...,d_j) \mid c\in \C, d=q(c)\}\subset\{0,1\}^{n+j}.$$ For $i=1,...,n$, let $$\C_{m+i} = \{ (d_1,...,d_m,c_1,...,c_{n-i+1}) \mid c\in \C, d=q(c)\}\subset\{0,1\}^{m+n-i+1}.$$ Finally, define $\C_{m+n+1} = q(\C)\subset \D$.

Now, for $j=1,...,m$, let the code map $q_j:\C_{j-1}\rightarrow \C_j$ be defined for $v=(c_1,...,c_n,d_1,...,d_{j-1})\in \C_{j-1}$ by  $q_j(v) = (c_1,...,c_n, d_1,...,d_j)\in \C_j$.   %Note that if $v=(c_1,...,c_n, d_1,...,d_{j-1})$, then $f_j(v)= f_j(c)$, as only the first $n$ places matter.  
Thus, if $v=(c_1,...,c_n,d_1,...,d_{j-1})$ with $d=q(c)$, then $q_j(v) = (c_1,...,c_n, d_1,...,d_j)$. Since $\phi$ is a neural ring homomorphism, the associated code map $q$ has a corresponding key vector $V$; note that $q_j$ is described by the key vector $W^j=(1,...,n+j-1, V_j)$, so $q_j$ is either repeating a neuron, or adding a trivial neuron, depending on whether $V_j = i$, or one of $u,0$.

Next, take the permutation map given by $\sigma = (n+1,...,n+m,1,...,n)$, so all the newly added neurons are at the beginning and all the originals are at the end.  That is,  define $q_\sigma:\C_m\rightarrow \C_{m+1}$  so if $v=(v_1,...,v_{n+m}),$ then $q_\sigma(v) = (v_{n+1},...,v_{n+m},v_1,...,v_n)$. 

We then delete the neurons $m+1$ through $n+m$ one by one in $n$ code maps.   That is, for $i=1,...,n$ define $q_{m+i}:\C_{m+i}\rightarrow \C_{m+i+1}$ by $q_{m+i}(v) = (v_1,...,v_{m+n-i})$.

Lastly, if $q(\C)\subsetneq \D$, then add one last inclusion code map $q_a:q(\C)\hookrightarrow \D$ to add the remaining codewords of $\D$.

Thus, given $c=(c_1,...,c_n)$ with $q(c) = d =(d_1,...,d_m)$, the first $m$ steps give us $q_m\circ\cdots\circ q_1(c) =  (c_1,...,c_n,d_1,...,d_m) = x$. The permutation then gives us $q_\sigma(x) = (d_1,...,d_m,c_1,...,c_n) = y$, and then we compose $q_{m+n}\circ\cdots\circ q_{m+1}(y) = (d_1,...,d_n) = d = q(c)$.   Finally, if $q(\C)\subsetneq \D$, we do our inclusion map, but as $q_a(d) = d$, the overall composition is a map $\C\rightarrow \D$ taking $c$ to $q_\phi(c)=d$ as desired. At each step, the map we use is from our approved list.

Finally, to show that code maps corresponding to neural ring isomorphisms only use maps (1)-(3), note that in the case that $\phi$ is a neural ring isomorphism, it is in particular an isomorphism, so the corresponding code map $q_\phi$ is a bijection and thus $q_\phi(\C) = \D$; no inclusion map is necessary in the last step of the process described above. We have also noted above that projection maps correspond to neural ring isomorphisms only when the deleted neuron is either trivial or a duplicate of another. Thus, only maps (1)-(3) are necessary to describe all neural ring isomorphisms. \end{proof}

\section{Neural ring homomorphisms and convexity}\label{sec:convexity}

One of the questions which has motivated a deeper understanding of the neural ring is that of determining which neural codes are convex.

\begin{definition} A neural code $\C$ on $n$ neurons is {\it convex in dimension $d$} if there is a collection $\mathcal{U} = \{U_1,...,U_n\}$ of convex open sets in $\R^d$ such that $\C = \{c\in \{0,1\}^n \,|\, \left(\bigcap_{c_i=1} U_i \right)\backslash \left(\bigcup_{c_j=0} U_j \right) \neq \emptyset\}$. If additionally no such collection exists in $\R^{d-1}$, then $d$ is known as the minimal embedding dimension of the code, denoted $d(\C)$. If there is no dimension $d$ where $\C$ is convex, then $\C$ is a {\it non-convex} code; in this case we use the convention $d(\C)=\infty$.
\end{definition}

\begin{example} In Example \ref{ex:fivemaps} (illustrated in Figure \ref{fig:codes}), we showed the results of applying five elementary code maps to the code $\C$. In that case, code $\C$ and its images $\C_1-\C_3$ are convex codes of dimension 2 and code $\C_4$ is convex of dimension 1. On the other hand,  $\C_5$ cannot be realized with convex sets in any dimension, as $U_1\cap U_2$ and $U_1\cap U_3$ necessarily form a disconnection of $U_1$.
\end{example}

In general, determining whether or not a code has a convex realization is a difficult question. Some partial results exist that give guarantees of convexity or of non-convexity, or that bound the embedding dimension (see for example \cite{neuralring, MRCalgsigs, MRCpaper, ShiuREU, ChadVlad, OpenClosedCodes}). One way to extend such results is to show that once a code is known to have certain properties related to convexity, we can generate other codes from it via code maps that would preserve these properties. The following theorem shows that if a {\it surjective} code map is `nice' (i.e., has a corresponding neural ring homomorphism), then it preserves convexity and the embedding dimension can only decrease.

\begin{theorem}\label{thm:convexity}
Let $\C$ be a code containing the all-zeros codeword and $q:\C\rightarrow\D$  a surjective code map corresponding to a neural ring homomorphism. Then if $\C$ is convex, $\D$ is also convex with $d(\D)\leq d(\C)$; if $\D$ is not convex, then $\C$ is not convex.
\end{theorem}

\begin{corollary}\label{cor:convexity}
Let  $\C$ be a code containing the all zeros codeword, and $q:\C\rightarrow \D$  a code map corresponding to a neural ring isomorphism. Then $\C$ and $\D$ are either both convex, with $d(\C) = d(\D)$, or both not convex. 
\end{corollary}

The proof of this theorem and its corollary relies on Theorem \ref{thm:mainthm}, and in particular uses the decomposition of these code maps to reduce the convexity question to code maps of just the five elementary types. As Theorem \ref{thm:convexity} addresses all neural ring homomorphisms that correspond to surjective code maps, it covers any such maps that are composed of permutation, duplication, deletion, or adding on trivial neurons.

 Note that the theorem would not necessarily hold for arbitrary surjective code maps that do {\it not} correspond to a neural ring homomorphism.  It would be a simple matter to create a bijection between a non-convex and a convex code with the same number of codewords, which would correspond to a ring isomorphism, but would not preserve convexity. 

The only non-surjective elementary code map corresponding to a neural ring homomorphism is inclusion, and this theorem cannot generally be extended to inclusion maps. Because the inclusion map can be used to include codes into arbitrary larger ones of the same length, it is possible to change convexity and dimension in arbitrary ways. The following examples show how to include convex codes in non-convex codes and vice versa, as well as ways to change the realization dimension by an arbitrary amount. 

\begin{example} Note that in Examples (1) and (3) below, we rely on results and constructions detailed in other work, especially \cite{MRCpaper}. 
\begin{enumerate}

\item Non-convex codes can be included into convex codes. If $\C$ is any non-convex code, then we can include $\C$ into the larger code $\Delta(\C)$, the simplicial complex of $\C$, which is necessarily convex. For more details, see for example \cite{MRCpaper, ChadVlad}.

\item Convex codes (of arbitrary dimension) can also be included into non-convex codes. Let $\C_1$ be a convex code on $n$ neurons, and $\C_2$ a non-convex code on $m$ neurons. Define the code $\C$ to be the code $\C_1$ with $m$ always-zero neurons appended to the end of each codeword; note that $\C$ is still convex, by the arguments above. Similarly, define the code $\C'$ to be the code $\C_2$ with $n$ always-zero neurons appended to the beginning of each codeword. The code $\C'$ is still not convex, again by the previous theorem. Define the code $\D$ to be the code $\C\cup \C'$, and note that as the first $n$ neurons never interact with the last $m$, this code is not convex, but we can include $\C$ into $\D$. 

\item Even when we include one convex code into another convex code, examples exist that change the dimension arbitrarily far in either direction. Let $n>2$ be arbitrarily large. Then, $\C = \{0,1\}^n \backslash \{11...1\}$ (the code on $n$ neurons with the all-ones codeword removed) is convex of dimension $n-1$. We can include $\C$ into the code $\D = \{0,1\}^n$, which is convex of dimension 2, reducing the dimension by $n-3$. We can also increase the dimension as far as we wish,  for example by including the simple one-dimensional code $\{00...0, 10...0\}$ into the code $\C = \{0,1\}^n\backslash \{11...1\}$, which was convex of dimension $n-1$, increasing the dimension by $n-2$. For a further discussion of the dimension and convexity of these codes, see \cite{MRCpaper}.
\end{enumerate}

\end{example}

We now give the proof of Theorem \ref{thm:convexity} and Corollary \ref{cor:convexity}.

\begin{proof}[Proof of Theorem \ref{thm:convexity}] 
If $\C$ is a surjective code map corresponding to a neural ring homomorphism, then it can be written as a composition of just the first four maps described by Theorem \ref{thm:mainthm}, following the process outlined in the proof.  Thus, to prove both theorem and corollary, it suffices to show that if a code $\C'$ is obtained from $\C$ via a projection map, then $d(\C')\leq d(\C)$, and that if $\C'$ is obtained from $\C$ via one of the first three maps, then $d(\C') = d(\C)$. In general, if a convex realization of $\C$ can be transformed, in the same dimension, into a convex realization for $\C'$, then we have shown both that $\C'$ is convex whenever $\C$ is, and also that $d(\C')\leq d(\C)$.

{\bf Permutation maps:} If $\C'$ is obtained from $\C$ via a permutation map, then any convex realization $\U$ of $\C$ is also a realization of $\C'$ by permuting the labels on the sets accordingly. Likewise, any realization $\U'$ of $\C'$ is a realization of $\C$, by permuting the labels inversely. Thus, $\C$ is convex if and only if $\C'$ is also convex, and in addition $d(\C') = d(\C)$.

{\bf Adding/deleting a trivial neuron:} If $\C'$ is obtained from $\C$ by adding a trivial always-zero neuron $n+1$, then a realization $\U$ of $\C$ can be transformed into a realization of $\C'$ by adding a set $U_{n+1}=\emptyset$. Likewise, a convex realization $\U'$ of $\C'$ can be transformed into a convex realization of $\C$ by removing the set $U_{n+1}$, which is necessarily empty as neuron $n+1$ never fires. For the second case, if $\C'$ is obtained from $\C$ by adding a trivial always-one neuron $n+1$, then we can transform a realization $\U$ of $\C$ into a realization of $\C'$ by adding the set $U_{n+1}$ that is made up of the entire ambient space $X$ in which the realization is set. This ambient space may be assumed to be convex, as $\C$ contains the all-zeros codeword.  Likewise, a realization of $\C'$ can be transformed to that for $\C$ by removing the set $U_{n+1}$. Thus, for such maps, $\C$ is convex if and only if $\C'$ is convex and, in addition, $d(\C)=d(\C')$. 

{\bf Adding/deleting a duplicate neuron:} if $\C'$ is obtained from $\C$ by duplicating neuron $i$ to a new neuron $n+1$, then any convex realization $\U$ of $\C$ can be transformed into a convex realization of $\C'$ by adding a set $U_{n+1}$ that is identical to the set $U_i$. Likewise, any convex realization $\U'$ of $\C'$ can also realize $\C$, by removing the set $U_{n+1}$ that must be identical to $U_i$.  Since $\C$ is obtained from $\C'$ by deleting a duplicate neuron, this argument also works for deleting a duplicate neuron. Hence, under such maps,  $\C$ is convex if and only if $\C'$ is convex, and in addition,  $d(\C) = d(\C')$.

{\bf Projection (deletion) maps:} if $\C'$ is obtained from $\C$ by deleting neuron $n$, then a convex realization $\U$ of $\C$ can be transformed into a realization of $\C'$ by removing the set $U_n$ from the realization. Thus, if $\C$ is convex, then $\C'$ must also be convex, and in particular, $d(\C')\leq d(\C)$. 
\end{proof}

\section*{Acknowledgements}

CC was supported by NIH R01 EB022862 and NSF DMS-1516881, NSF DMS-1225666/1537228, and an Alfred P. Sloan Research Fellowship; NY was supported by the Clare Boothe Luce Foundation.  The authors thank Katie Morrison, Mohamed Omar, and  R. Amzi Jeffs for many helpful discussions.

%%%%BIBLIOGRAPHY%%%%%%

\bibliographystyle{unsrt}
\bibliography{maps-between-codes}

\end{document}